\documentclass[journal]{IEEEtran}

\makeatletter
\def\ps@headings{%
\def\@oddhead{\mbox{}\scriptsize\rightmark \hfil \thepage}%
\def\@evenhead{\scriptsize\thepage \hfil \leftmark\mbox{}}%
\def\@oddfoot{}%
\def\@evenfoot{}}
\makeatother
\usepackage{cite}
\usepackage{epstopdf}
\usepackage{graphicx}
\usepackage{algorithm,algpseudocode}
\usepackage[cmex10]{amsmath}
\usepackage{ifthen}
\usepackage[margin=10pt,font=small,labelfont=bf]{caption}
\usepackage{subcaption}

\newboolean{conference}
\setboolean{conference}{false}

\graphicspath{{figures/}}

\def\beq{\begin{equation}}
\def\eeq{\end{equation}}
\def\beqa{\begin{eqnarray}}
\def\eeqa{\end{eqnarray}}
\def\beqan{\begin{eqnarray*}}
\def\eeqan{\end{eqnarray*}}

\def\argmax{\mathop{\mathrm{arg\,max}}}

\def\x{\times}

\newtheorem{theorem}{Theorem}
\newtheorem{lemma}{Lemma}

\def\xhat{\widehat{x}}

\def\la{\leftarrow}
\def\ra{\rightarrow}

\def\arr{\rightarrow}

\def\tm1{t\! - \! 1}
\def\tp1{t\! + \! 1}

\def\Xset{{\cal X}}

\def\xbf{\mathbf{x}}
\def\xbfhat{\widehat{\mathbf{x}}}

\def\Rbar{\overline{R}}

\begin{document}
\bibliographystyle{IEEEtran}

\title{Wireless Scheduling with Dominant Interferers and Applications
to Femtocellular Interference Cancellation}
\ifthenelse{\boolean{conference}}
{ 

    \author{
        \IEEEauthorblockN{Mustafa Riza Akdeniz}
        \IEEEauthorblockA{Polytechnic Institute of\\New York University\\
        Brooklyn, New York\\
        Email: makden01@students.poly.edu}
        \and
        \IEEEauthorblockN{Sundeep Rangan}
        \IEEEauthorblockA{Polytechnic Institute of\\New York University\\
        Brooklyn, New York\\
        Email: srangan@poly.edu}
    }
}{ 
    \author{
        Mustafa Riza Akdeniz,~\IEEEmembership{Student~Member,~IEEE},
        Sundeep Rangan,~\IEEEmembership{Member,~IEEE},
        \thanks{This material is based upon work supported by the National Science
        Foundation under Grant No. 1116589.}
        \thanks{M. Akdeniz (email:makden01@students.poly.edu) and
                S. Rangan (email: srangan@poly.edu) are with the
                Polytechnic Institute of New York University, Brooklyn, NY.}
    }
}

\maketitle

\begin{abstract}
We consider a general class of wireless scheduling and resource allocation problems where
the received rate in each link is determined by the actions of the transmitter
in that link along with a single dominant interferer.  Such scenarios arise in a range
of scenarios, particularly in emerging femto- and picocellular networks with
strong, localized interference.  For these networks, a utility maximizing scheduler based
on loopy belief propagation is presented that enables computationally-efficient
local processing and low  communication overhead.
Our main theoretical result shows that the fixed points of the method
are provably globally optimal for arbitrary
(potentially non-convex) rate and utility functions.
The methodology thus provides globally optimal solutions to
a  large class of inter-cellular interference coordination problems including
subband scheduling, dynamic orthogonalization and beamforming whenever the dominant
interferer assumption is valid.  The paper focuses on applications for systems with
interference cancellation (IC) and suggests a new scheme on optimal rate control,
as opposed to traditional power control.  Simulations are presented in industry
standard femtocellular
network models demonstrate significant improvements in rates over simple reuse 1 without
IC, and near optimal performance of loopy belief propagation for rate selection
in only one or two iterations.
\end{abstract}

\ifthenelse{\boolean{conference}}{}{
    \begin{IEEEkeywords}
    femtocells, cellular systems, belief propagation, message passing, interference cancellation.
    \end{IEEEkeywords}
}

\section{Introduction}
\ifthenelse{\boolean{conference}}{A}{\IEEEPARstart{A}{}}
central challenge in next-generation cellular networks is the presence of strong
interference.  This feature is particularly prominent
in emerging femto- and pico-cellular networks where, due to \emph{ad hoc} and
unplanned deployments, restricted association and mobility in
small cell geometries, mobiles may be exposed to much higher levels of
interference than in traditional planned macrocellular networks
\cite{ChaAndG:08,LopezVRZ:09,AndrewsCDRC:12}.
Methods for advanced intercellular interference coordination (ICIC)
have thus been a key focus of 3GPP LTE-Advanced standardization efforts~\cite{3GPPICIC}
and other cellular standards organizations~\cite{FemtoForum:10}.

This work considers a class of general resource
allocation and scheduling problems
where the received rate in each link is determined primarily by the actions of the
transmitter in that link along with a single dominant interferer.
Interference in traditional macrocellular networks, of course, generally do \emph{not}
follow this assumption as there are typically a large number of mobiles per cell,
and interference generally arises from an aggregate of signals from many sources in the network.
However, in emerging small cell networks with higher cell densities,
the number of mobiles per cell is significantly reduced and strong
interference is most often due to the poor placement of the mobile relative to a single
victim or interfering cell to which the mobile is not connected to.

For such dominant interference networks, we propose a novel ICIC
algorithm based on graphical models and message passing.
Graphical models \cite{WainwrightJ:08,Frey:98}
are a widely-used tool for high-dimensional optimization and Bayesian inference problems
applicable whenever the objective function or posterior distribution
factors into terms with small numbers of variables.
Methods such as loopy belief propagation (BP) then
reduce the global problem to a sequence of
smaller local problems associated with each factor.  The outputs of the local
problems are combined via message passing.
The loopy BP methodology is particularly well-suited to wireless scheduling problems since
the resulting algorithms are inherently distributed and require only
local computations.  Indeed, \cite{Chiang:02} showed that
many widely-used network routing,  congestion control and
power control algorithms can be interpreted as instances
of the sum-product variant of loopy BP.  More recently,
\cite{SanghaviMW:07, BayatiSS:08} used BP techniques for networks with contention graphs,
\cite{sohn_2010} proposed BP for MIMO systems,
and \cite{RanganM:11,RanganM:12} considered a general class of
ICIC problems with weak interference exploiting an approximate BP technique in
\cite{RanganFGS:12-ISIT}.

For dominant interferer networks considered in this paper,
we apply loopy BP methods to a general network utility
maximization problem where the goal is to maximize a sum of utilities
across a system with $n$ links.  The utility in each link is assumed to be
a function of a scheduling decision
on the transmitter of that link, as well as the decision of at most \emph{one}
interfering link.
Aside from this single dominant interferer assumption, the model is completely general,
and can incorporate arbitrary channel relationships as well as sophisticated transmission
schemes including subband scheduling, beamforming and dynamic orthogonalization --
key features of advanced cellular technologies such as 3GPP LTE \cite{Dahlman:07}.
In addition, computing maximum weighted matching for queue
stability~\cite{TassiulasE:92} and maximization of sum utility of
average rates for fairness~\cite{stolyar_2005} can be incorporated into
this formulation.

We show that for single dominant interferer networks,
loopy BP admits a particularly simple implementation:
scheduling decisions can be computed with a small number of rounds
of message passing, where each
round involves simple local computations at the nodes along with communication only
between the transmitter, receiver and dominant interferers in each link.
Our simulations in commercial cellular models indicate near optimal performance in
very small, sometimes only one or two, rounds -- making the algorithms extremely attractive for
practical implementations.

Moreover, we establish rigorously (Theorem \ref{thm:fixPoint})
that any fixed point of the loopy BP algorithm
is guaranteed to be \emph{globally} optimal.  Remarkably, this result
relies solely on the dominant interferer
assumption and otherwise applies to arbitrary utility functions and relationships
between scheduling decisions on interfering and victim links.
In particular, this optimality holds even for non-convex problems.
Optimality results for loopy BP are generally difficult to establish:
Aside from graphs that are cycle free, loopy BP generally only provides approximate
solutions that may not be globally optimal.  Our proof of global optimality
in this case rests on showing that
graphs with single dominant interferers have at most cycle in each connected component.
The result then follows from a well-known optimality property in \cite{WeissFree:01}.

\subsection{Interference Cancelation and Rate Control}

A target application of our methodology is for scheduling problems
with \emph{interference cancelation} (IC) \cite{Ahlswede:71,CoverT:91}.
When an interfering signal is sufficiently strong, it can be decoded
and canceled prior to or jointly with the decoding of the desired signal,
thus eliminating the interference entirely.
IC, particularly when used in conjunction with techniques such as rate splitting
\cite{HanKob:81}, is known to provide significant performance gains
in scenarios with strong interference, where it may even be optimal or near optimal
\cite{EtkinTseWa:08}.  IC may also be useful for mitigating cross-tier
 interference between short-range (e.g.\ femtocellular)
  links operating below larger macrocells \cite{Rangan:10-GC}.
However, while improved computational resources have recently made IC implementable
in practical receiver circuits \cite{Andrews:05,BourdreauPGCWV:09,ShiReed:07},
it is an open problem of how rates should be selected in larger networks
when IC is available at the link-layer.

In this paper, we consider a network where the receiver in each link
can jointly detect the desired signal along with at most one interfering link.
The limitation of joint detection with at most one interfering link
is reasonable since the likelihood of having two very strong interferers is low
in most practical scenarios.
The limit is also desirable since
the computationally complexity at the receiver grows significantly with
each additional link to perform joint detection with.
For networks with IC, we propose to perform  interference coordination based on \emph{rate control}
rather than traditional power control that has been the dominant method in cellular
systems without IC \cite{ChiangHLT:08}.  Specifically, each transmitter operates at
a fixed power, and the system utility is controlled by rate:
increasing the rate improves the utility to the desired user,
while decreasing the rate makes the transmission more ``decodable"
and hence ``cancelable" at the receiver
in any victim link.  Hence, we suggest that in networks with IC, rate control,
as opposed to power control, can be viable method for interference coordination.

\section{System Model}

We consider a system with $n$ links, each link $i$ having one transmitter, TX$i$ and
one receiver, RX$i$.  In a cellular system, multiple logical links may be associated with the
base station, either in the downlink or uplink.  Each link $i$ is to make some
\emph{scheduling decision}, meaning a selection of some variable $x_i \in \Xset_i$
for some set $\Xset_i$.  The scheduling decision could include choices, for example,
of power, beamforming directions, rate or vectors of these quantities in the cases
of multiple subbands -- the model is general.
Associated with each RX$i$ is a \emph{utility function} representing some value or quality of
service obtained by RX$i$.  We assume that the utility function has the form
\[
    f_i(x_i,x_{\sigma(i)}),
\]
where $\sigma(i) \in \{1,\ldots,n\}$ is the index of one link other than $i$,
representing the index of a \emph{dominant} interferer to RX$i$.
Thus, the assumption is that
utility on each link $i$ is a function of the scheduling decisions of the serving transmitter,
TX$i$, and one dominant interfering transmitter, TX$j$ for $j=\sigma(i)$.
Generally,
the dominant interferer with be the transmitter, other than the serving transmitter,
with the lowest path loss to the receiver.
The problem is to find the optimal solution
\beq \label{eq:xopt}
    \xbfhat := \argmax_{\xbf} F(\xbf),
\eeq
where the objective function is the sum utility,
\beq \label{eq:Fsum}
    F(\xbf) := \sum_{i=1}^n f_i(x_i,x_{\sigma(i)}).
\eeq

The utility function $f(x_i,x_{\sigma(i)})$ accounts for the
link-layer conditions that determine the rates and along with quality-of-service (QoS)
requirements for valuation of the traffic.
A general treatment of utility functions for scheduling problems
can be found in \cite{KelleyMT:98,ShakkottaiS:07}.
The optimization problem \eqref{eq:xopt} can be applied to both static and dynamic
problems.

For static optimization, the scheduling vectors $x_i$
are selected once for a long time period
and the utility function is typically of the form
\beq \label{eq:utilStatic}
    f_i(x_i,x_{\sigma(i)}) = U_i(R_i(x_i,x_{\sigma(i)})),
\eeq
where $R_i(x_i,x_{\sigma(i)})$ is the long-term rate as a function of the
TX scheduling decision $x_i$ and decision $x_{\sigma(i)}$ on the dominant
 interferer, while $U_i(R)$ is the utility as a function of the rate.
The problem formulation above can incorporate any of the
common utility functions including: $U_i(R) = R$ which
results in a sum rate optimization; $U_i(R) = \log(R)$ which is the proportional
fair metric and $U_i(R) = -\beta R^{-\beta}$ for some $\beta > 0$ called
an $\beta$-fair utility.  Penalties can also be added if there is a cost
associated with the selection of the TX vector $x_i$ such as power.

To accommodate time-varying channels and traffic loads,
many cellular systems enable fast dynamic scheduling in time slots
in the order of 1 to 2 ms.  For these systems, the utility maximization can be
re-run in each time slot.  One common approach is that
in each time slot $t=0,1,2\ldots$,
the scheduler uses a utility of the form
\beq \label{eq:wtUtil}
   f_i(t,x_i(t),x_{\sigma(i)}(t))= w_i(t)R_i(t,(x_i(t),x_{\sigma(i)}(t)), \ \ \
\eeq
where $w_i(t)$ is a time-varying weight given by the marginal utility
\beq \label{eq:wtMargUtil}
   w_i(t) = \frac{\partial U_i(\Rbar_i(t))}{\partial R},
\eeq and $\Rbar_i(t)$ is exponentially weighted average rate updated
as \beq \label{eq:rateFilt}
    \Rbar_i(t+1) = (1-\alpha)\Rbar_i(t) + \alpha
    R_i(t,\xhat_i(t),\xhat_{\sigma(i)}(t)),
\eeq
where $\xhat_i(t)$ and $\xhat_{\sigma(i)}$ are the TX scheduling decisions
on the serving and interfering links at time $t$. Any maxima of the
optimization \eqref{eq:xopt} with the weighted utility
\eqref{eq:wtUtil} is called a \emph{maximal weight matching}. A
well-known result of stochastic approximation~\cite{stolyar_2005} is
that if $\alpha \arr 0$, and the scheduler performs the maximum
weight matching with the marginal utilities \eqref{eq:wtMargUtil},
then for a large class of processes, the resulting average rates
will maximize the total utility $\sum_i U_i(\Rbar_i(t))$.

The above utilities are designed for infinite backlog queues.
For delay sensitive traffic, one can take the weights $w_i(t)$ to be the queue length
or head-of-line delay.  Maximal weight matching performed with these weights
generally results in so-called throughput optimal performance \cite{TassiulasE:92}.
These results also apply to multihop networks with the so-called backpressure
weights.

\subsection{Systems with IC}

As mentioned in the Introduction, a particularly valuable application of the methodology
in this paper is for rate selection in systems with interference cancelation (IC).
To simplify the exposition, we consider the following simple model of a system with IC:
Assume the power of each transmitter, TX$j$, is fixed to some
level $P_j$, and let $G_{ij}$ denote the gain from TX$j$ to RX$i$ as let $N_i$
denote the thermal noise at RX$i$.
Without IC, the receiver RX$i$ would experience
a signal-to-interference-and-noise ratio (SINR) given by
\beq \label{eq:sinrReuse}
    \rho_i := \frac{G_{ii}P_i}{\sum_{k \neq i} G_{ik}P_k + N_i}.
\eeq
If the interference is treated as Gaussian noise, 
and the system were to operate at the Shannon capacity, the set of rates $R_i$ attainable
at RX$i$ would be limited to
\beq \label{eq:cap1}
    R_i \leq \log_2(1+\rho_i).
\eeq
We call this set of achievable rates the \emph{reuse 1 region}, since
this set is precisely the rates achievable in a cellular system operating
with frequency reuse 1 (i.e.\ all transmitters transmitting across the entire bandwidth)
and treating interference as noise.
We denote the reuse 1 region by ${\mathcal C}_{\rm reuse1}(i)$.

The addition of IC can be seen as a method to expand this rate region.
Specifically,
suppose each RX$i$ can potentially jointly detect and cancel the signals from
at most one interfering transmitter TX$j$ for some index $j=\sigma(i) \neq i$.
The limitation to a single interferer will enable the system
to fit within the dominant interferer model described above.
However, this limitation has minimal
practical impact since jointly detecting more than one interferer is
both computationally intensive and seldom of much value anyway since the presence of strong
interference from more than one source is rare.

Now, to compute the region achievable with joint detection and IC consider Fig.~\ref{fig:region}
which shows the receiver RX$i$ being served by the transmitter TX$i$
while receiving interference from TX$j$ for $j=\sigma(i)$.
In this model, assuming again that RX$i$ can operate at the Shannon capacity, 
it can jointly detect the signals from both
the serving transmitter TX$i$ and interferer TX$j$, if and only if
the rates $R_i$ and $R_j$ satisfy the multiple access channel (MAC)
conditions \cite{CoverT:91}:
\begin{subequations} \label{eq:cap2}
\beqa
    R_i &\leq& \log_2 (1+\tilde{\rho}_{i,i})\\
    R_{\sigma(i)} &\leq& \log_2 (1+\tilde{\rho}_{i,\sigma(i)})\\
    R_i + R_{\sigma(i)} &\leq& \log_2(1+\tilde{\rho}_{i,i}+\tilde{\rho}_{i,\sigma(i)}),
\eeqa
\end{subequations}
where, for $\ell=i$ or $\ell=\sigma(i)$,
$\tilde{\rho}_{i,\ell}$ is the SINR,
\beq \label{eq:sinrIC}
    \tilde{\rho}_{i,\ell} = \frac{G_{i\ell}P_\ell}{\sum_{k \neq i,\sigma(i)}G_{ik}P_k+N_i}.
\eeq
We denote this region by ${\mathcal C}_{\rm IC}(i)$ which is plotted in the bottom panel
of Fig.~\ref{fig:region} along with reuse 1 region ${\mathcal C}_{\rm reuse1}(i)$.
We can see that the addition of IC expands the set of achievable in the region
where the interfering rate is low.  In particular, when the interfering rate $R_{\sigma(i)}$
is sufficiently low (to the left of point A in the figure), the achievable rate $R_i$
on the serving link is identical to the case without any interference.
In this sense, reducing the rate makes it more ``decodable" or ``cancelable".
We let ${\mathcal C}(i)$ denote the total set of feasible rates, which is the union of
the reuse 1 and IC regions:
\[
    {\mathcal C}(i) =     {\mathcal C}_{\rm IC}(i) \cup  {\mathcal C}_{\rm reuse1}(i).
\]
Note that this region is \emph{not} convex.

With these observations, we can now pose the IC problem in the dominant interferer model
as follows:  The decision variable at each transmitter TX$i$ is an \emph{attempted}
rate, denoted $x_i$.  The \emph{achieved} rate at the receiver RX$i$ is then
simply
\[
    R_i = R_i(x_i,x_{\sigma(i)}) := \left\{ \begin{array}{ll}
        x_i, & \mbox{if }
            (x_i,x_{\sigma(i)}) \in {\mathcal C}(i) \\
        0, & \mbox{if }
            (x_i,x_{\sigma(i)}) \not \in {\mathcal C}(i)
            \end{array} \right.
\]
That is, the rate is achieved rate is equal to the attempted rate on the serving link
if and only if the serving and interfering rates are feasible; otherwise, the achieved
rate is zero.
Then, the optimization can be formulated as \eqref{eq:utilStatic} for any utility
function $U_i(R_i)$.

\begin{figure}
    \centering
    \includegraphics[width=0.3\textwidth]{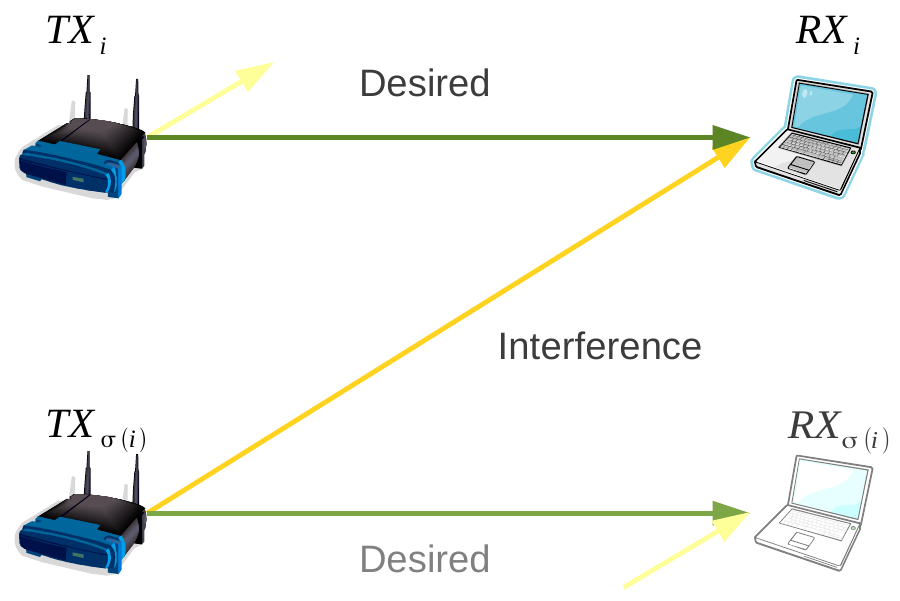} \\
    \includegraphics[width=0.4\textwidth]{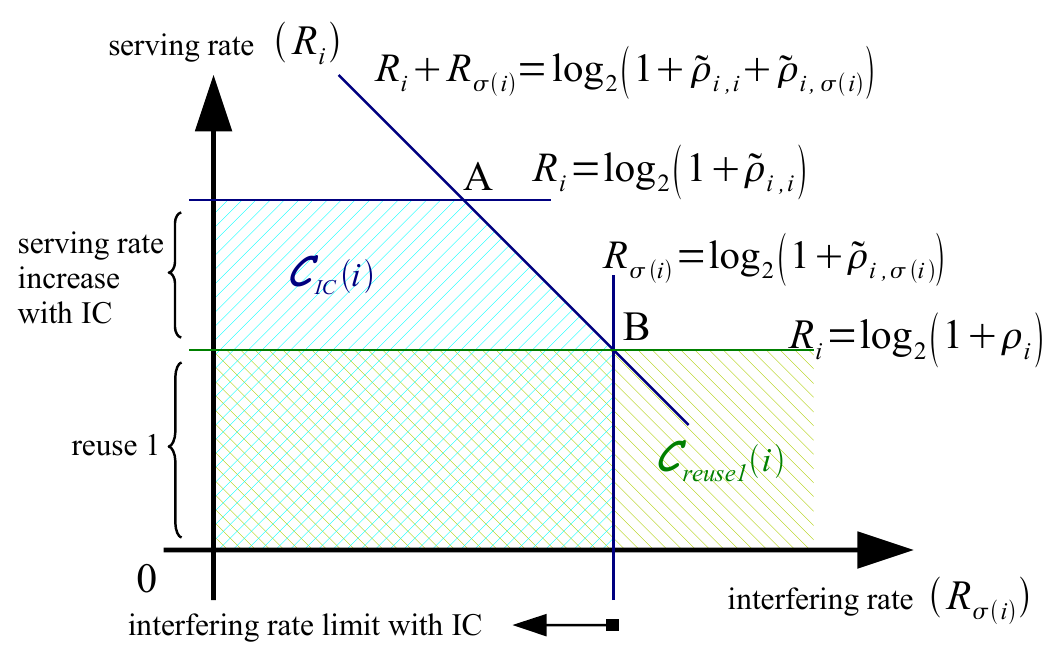}

  \caption{\textbf{Rate region with IC}:  When a receiver RX$i$ is capable of IC,
  its achievable rate depends on both the serving rate rate $R_i$ as well as the
  rate $R_{\sigma(i)}$ on the interfering link. The total rate region, which includes
  the options of both reuse 1 and treating interference as noise as well as joint detection,
  is non-convex. }
  \label{fig:region}
\end{figure}

More sophisticated methods, such as power control and rate splitting as used in the
well-known Han-Kobayashi (HK) method \cite{HanKob:81}, can also be incorporated into 
this methodology.  For example, to incorporate the HK technique, the transmitter TX$i$
would splits its transmissions into ``private" and ``public" parts and the decision
variable $x_i$ for that link would be a vector including the private and public rates
and their power allocations.  However, for simplicity, we do not simulate this scenario.

\section{Message Passing Algorithms}

\subsection{Graphical Model Formulation}

As mentioned in the Introduction, graphical models provide a general and systematic
approach for distributed optimization problems where the objective function admits
a factorization into terms each with small numbers of variables \cite{WainwrightJ:08,Frey:98}.
The optimization \eqref{eq:xopt} with
objective function \eqref{eq:Fsum} is ideally suited for this methodology.

To place the optimization problem into the graphical model formalism, we define
a \emph{factor} graph $G=(V,E)$, which is an undirected
bipartite graph whose vertices
consists of $n$ variable nodes associated with the decision variables $x_i$, $i=1,\ldots,n$,
and $n$ factor nodes associated with factors $f_i$, $i=1,\ldots,n$.
There is an edge between $x_\ell$ and $f_i$ if and only $x_\ell$ appears as an argument
in $f_i$ -- namely if $\ell=i$ or $\ell=\sigma(i)$.

As an example of the factor graph consider the network in the top panel of
Fig.~\ref{fig:netExample} with $n=6$ links.  Each link has a transmitter and receiver
and the diagram indicates both the serving and interfering links.
The factor graph representation is shown in the bottom panel where there is one
variable, $x_i$, and one factor, $f_i$, for each link $i=1,\ldots,n$.

\begin{figure}
\centering
\includegraphics[width=0.25\textwidth]{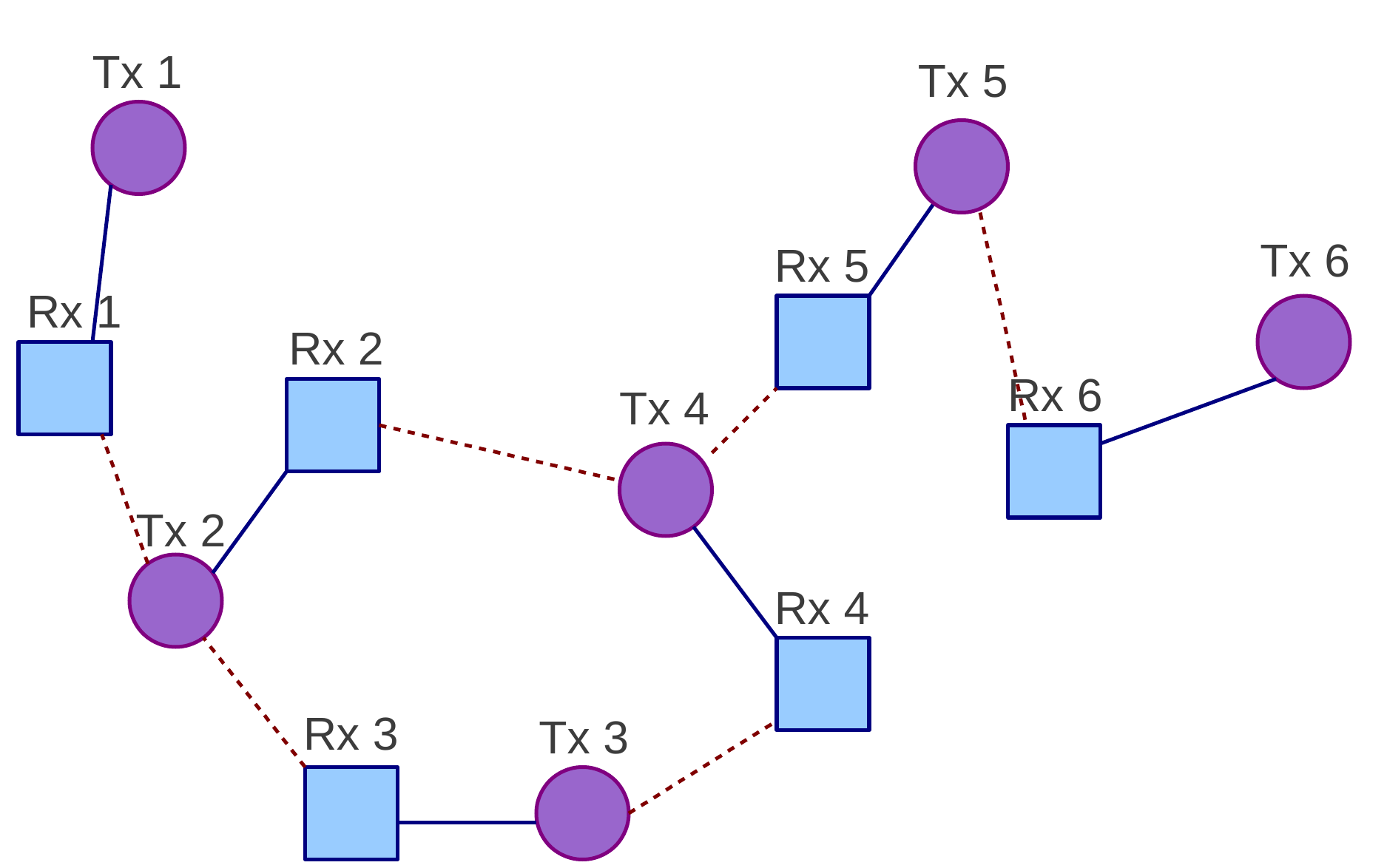} \\
\includegraphics[width=0.35\textwidth]{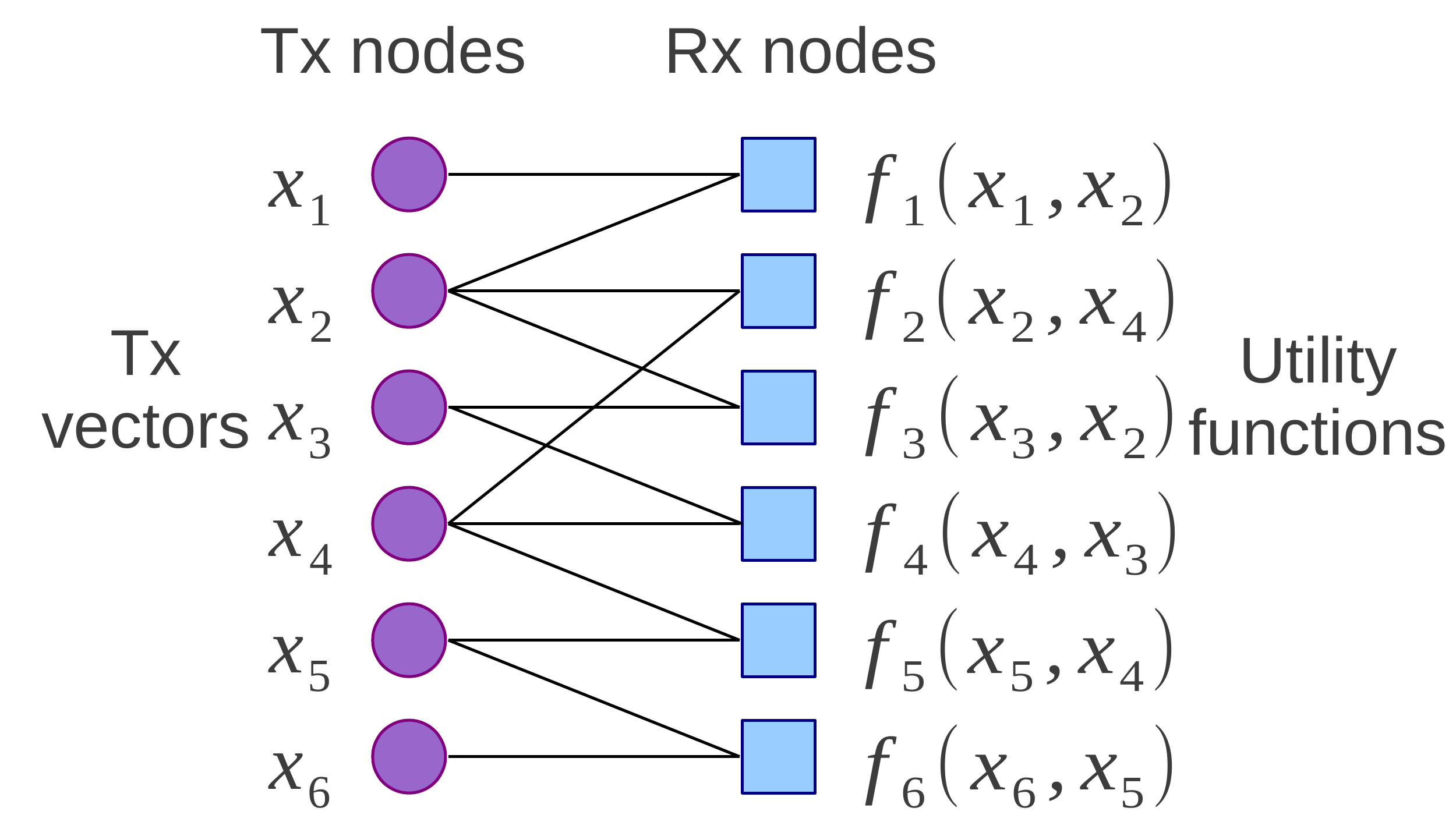}
\caption{\textbf{Example factor graph representation of a network:}
The top panel shows an example network with $n=6$ links, each with a transmitter and receiver.
Solid lines indicate serving links and dotted lines indicate interfering links.
The bottom panel shows its factor graph representation.}
\label{fig:netExample}
\end{figure}

\subsection{Max-Sum Loopy BP}
Once the optimization problem has been formulated as a graphical model, we can apply
the standard max-sum loopy BP algorithm.  A general description of the max-sum algorithm can be found
in any standard graphical models text such as \cite{WainwrightJ:08,Frey:98}.
For the dominant interferer optimization problem, the max-sum algorithm can be implemented
as shown in Algorithm \ref{algo:maxSum}.

Loopy BP is based on iteratively passing ``messages" between the variables and factor nodes.
For the interferer problem, the messages are passed between the receivers RX$i$ and
transmitters TX$j$ where TX$j$ is either the serving transmitter or dominant interferer
of RX$i$.

Each message is a function of the scheduling decision $x_j$.
The message function from  TX$j$ to RX$i$ is denoted
$\mu_{j \ra i}(x_j)$ and represents an estimate of the value of TX$j$ selecting
the scheduling decision $x_j$.  This message can thus be interpreted as a sort of ``soft"
request-to-send.  We use the term ``soft" since the message ascribes a value to each possible
scheduling decision.
Similarly, the messages from RX$i$ to TX$j$ is denoted $\mu_{j \la i}(x_j)$
and has a role has both a ``soft" clear-to-send and channel quality indicator (CQI).

\begin{algorithm}
\caption{Max-sum loopy BP for scheduling}
\label{algo:maxSum}
\begin{algorithmic}

    \State \{Initialization at the RX \}
    \ForAll {RX$i$}
        \State $\mu_{i \ra i}(x_i) \gets 0$
        \State $\mu_{j \ra i}(x_j) \gets 0$ for $j=\sigma(i)$
    \EndFor
    \Repeat

        \State \{ Receiver half round \}
        \ForAll{RX$i$}
            \State $j \gets \sigma(i)$
            \State $\mu_{i \la i}(x_i) \gets
                \max_{x_j} f_i(x_i,x_j)+ \mu_{j \ra i}(x_j)$
            \State $\mu_{j \la i}(x_j) \gets \max_{x_{i}}
                f_i(x_i,x_j)+\mu_{i \rightarrow i}(x_i)$
        \EndFor

        \State {}
        \State \{ Transmitter half round \}
        \ForAll {TX$j$}

            \State $H_j(x_j) \gets \mu_{j \la j}(x_j) +
                \sum_{i:j=\sigma(i)}\mu_{j \la i}(x_j)$
            \State $\mu_{j \ra j}(x_j) \gets  H_j(x_j) - \mu_{j \la j}(x_j)$
            \ForAll {$i$ s.t. $j=\sigma(i)$}
                \State $\mu_{j \ra i}(x_j) \gets  H_j(x_j) - \mu_{j \la i}(x_j)$
            \EndFor
      \EndFor
    \Until {max number of iterations}
    \State {}
    \State \{ Final scheduling decision \}
    \ForAll {TX$j$}
        \State $H_j(x_j) \gets \mu_{j \la j}(x_j) +
                \sum_{i:j=\sigma(i)}\mu_{j \la i}(x_j)$
        \State $\xhat_j \gets \argmax_{x_j} H_j(x_j)$
    \EndFor

\end{algorithmic}
\end{algorithm}

As shown in Algorithm \ref{algo:maxSum}, the messages at the receiver are initialized at zero
and then iteratively updated in a set of \emph{rounds}.  Each round consists of two
halves:  In the first half, the receivers send messages to the transmitters and,
in the second half, the transmitters send messages back to the receivers.
The process is a repeated for a fixed number of iterations, and the final messages are
used to compute the scheduling decision at the transmitters.

\subsection{Optimality under the Dominant Interferer Assumption}
We will discuss more detailed issues in implementing the algorithm momentarily.
But, we first establish an important optimality, which is the main justification for the algorithm.

\medskip
\begin{theorem} \label{thm:fixPoint}  Consider the message passing max-sum algorithm,
Algorithm \ref{algo:maxSum}, applied to the optimization \eqref{eq:xopt}
with the objective function \eqref{eq:Fsum} for \emph{any} utility functions $f_i(x_i,x_{\sigma(i)})$
and dominant interferer selection function $\sigma(i)$.
Then, if all messages $\mu_{i \ra j}(x_i)$ and $\mu_{i \la j}(x_i)$
 are fixed-points of the algorithm, the resulting scheduling decisions
 $\xhat_j$ are \emph{globally} optimal solutions to \eqref{eq:xopt} in that
 \[
    F(\xbfhat) \geq F(\xbf),
 \]
 for all other scheduling decision vectors $\xbf$.
\end{theorem}
\begin{proof}  See Appendix \ref{sec:proof}.
\end{proof}

\medskip
The theorem shows that if the max-sum algorithm converges to a fixed point,
the solution will be optimal.  Remarkably, this result applies to arbitrary utility functions
as long as the dominant interferer assumption is valid.  In particular, the result applies
to even non-convex utilities such as the ones arising in the IC problem.
The result is unexpected since loopy BP is generally only guaranteed to be optimal in
graphs with no cycles, and the factor graph we are considering may not be acyclic.

Of course, the theorem does \emph{not} imply the convergence of the algorithm.
Indeed, since the graph has cycles, the algorithm may not converge.  However,
the iterations of loopy can be ``slowed down" using fractional updates with the same
fixed points to improve convergence at the expense of increased number of iterations.
However, as we will see in the simulations, convergence does not appear to be an issue
in the cases we examine.

\subsection{Implementation Considerations}

We conclude with a brief discussion of some of the practical considerations in 
implementing the message-passing algorithm in commercial cellular systems.

\begin{itemize}
\item \emph{Communication overhead:}  Most important, it is necessary to recognize that
the message-passing needed for the proposed algorithm would require additional
control channels not present in current cellular standards such as 
LTE or UMTS.  New control channels would be needed to communicate the rounds of messaging
prior to each scheduling decision.  However, similar messages have been proposed in 
Qualcomm's peer-to-peer system, FlashLinQ \cite{FlashLinq:10-allerton},
and also appear in some optimized CSMA algorithms such as \cite{NiSrikant:09}.
Under the dominant interferer assumption,  the communication overhead may particularly small 
since each receiver must send
and receive messages to and from at most
two transmitters in each round (the serving transmitter and the dominant interferer).  
However, each message is an entire function,
so the message must in principle
contain a value for each possible $x_i$.  However, if the number of values $x_i$ is small,
or the function can be well-approximated in some simple parametric form, the overall communication
per round overhead will be low.  Moreover, as we will see in the simulations, very small numbers
of rounds are typically needed; sometimes as small as one or two.

\item \emph{Computational requirements:}  The computational requirements are low.
The transmitter must simply sum functions, and the receivers must perform one-dimensional
optimizations.  Thus, as long as the sets of possibilities for each $x_i$ is small, the
computation will be minimal.

\item \emph{Channel and queue state information:}
The utility $f_i(x_i,x_j)$ is generally based on the rate achievable at RX$i$
as a function of the scheduling decisions at both the serving transmitter TX$i$ and the
dominant interferer TX$j$.  Thus, for the receiver to perform the optimization
in the message passing algorithm, the receiver must generally know the channel state
from the both the serving and interfering channels along with noise from all other sources.
This channel state can be estimated from pilot or reference signals which are already
transmitted in most cellular systems.
However, the utility function will also generally depend on
the queue state at the transmitter, particularly,
the priorities and size of packets in the transmit buffer.
This information must thus be passed, somehow, to the receiver.  Cellular systems such as LTE
already pass such information in the uplink through buffer status reports \cite{Dahlman:07},
however a similar channel would need to be placed in the downlink for this message passing
algorithm to work.

\item \emph{Synchronization:}  Since the message passing loopy BP algorithm
is based on coordinated scheduling, there is an implicit assumption of synchronization.
In a system such as LTE, this would require synchronization at the subframe level,
where the subframes are 1ms.  Since the intended use case for such coordinated scheduling
applications is for small cell systems (e.g.\ 100 to 200m cell radii), such synchronization
accuracy is reasonable.

\end{itemize}

\section{Simulation Results}

We validate the max-sum loopy BP algorithm for interference cancellation
in two simulations.  In both cases, we consider static optimizations of log utility,
$U(R_i) = \log R_i$, which is the proportion fair metric \cite{KelleyMT:98}.
For the link-layer model, we assume that the
rate region is described by the Shannon capacity with a loss of 3dB and there is a maximum
spectral efficiency of 5 bps/Hz.  The maximum spectral efficiency in practice is determined
by the maximum constellation, which itself is usually limited by the receiver dynamic range.
In the implementation of loopy BP, we assume a discretization of
the rates using logarithmic spacing of the 25 rate points from zero to the maximum spectral efficiency.
In practice, one should use the discrete modulation and coding
scheme (MCS) values that are available at the link layer.  In LTE,
there are approximately 30 such MCS values \cite{3GPP36.211}.

Note that the final rates $\xhat_j$ from Algorithm \ref{algo:maxSum} may not be feasible,
unless the algorithm has been run to convergence.  To evaluate the algorithm
prior to convergence, we obtain a feasible point from the values $\xhat_j$
at the end of the algorithm as follows:
At each receiver RX$i$, if $(\xhat_i,\xhat_{\sigma(i)}) \in {\mathcal C}(i)$ (that is,
the rate pair is feasible), we take the rate $R_i=\xhat_i$.  Otherwise we take $R_i$ to
be the largest value with $R_i \leq \xhat_i$ and
$(R_i,\xhat_{\sigma(i)}) \in {\mathcal C}(i)$.  It is simple to show that
the resulting set of rates, $(R_i,R_{\sigma(i)})$, is feasible for all $i$.
This final projection step can be performed with local information only.

\subsection{3GPP Femtocellular Apartment Model}
In our first simulation, we use a simplified version of an industry standard model \cite{FemtoForum:10}
for validating ICIC algorithms of femtocells in densely packed apartments.
The apartments are assumed to be laid out in a 4x4 grid, each apartment being 10x10m.
In 10 of the 16 apartments, we assume there is an active link with
a femto base station and mobile (user equipment or UE in 3GPP terminology).
The femtos are assumed to operate in closed access mode so that a
femto UE can only connect to the femto base station in its apartment.
This restricted association is a major source of strong interference in femtocellular
networks and thus a good test case for IC.
The complete list of simulation parameters are given on the Table~\ref{tbl:simParam1}.

\begin{table}
  \begin{center}
    \begin{tabular}{|p{1in}|p{2in}|}
      \hline
      Parameter & Value \\ \hline
      Network topology & $4 \x 4$ apartment model,with active links in 10 of the 16 apartments. \\ \hline
      Bandwidth & 5 MHz \\ \hline
      Wall loss & 10 dB \\ \hline
      Lognormal shadowing & 10 dB std.\ dev.\ \\ \hline
      Path loss & $38.46 + 20\log_{10}(R) + 0.7R$ dB, $R$ distance in meters. \\ \hline
      Femto BS TX power & 0 dBm \\ \hline
      Femto UE noise figure & 4 dB \\ \hline
    \end{tabular}
  \end{center}
  \caption{Simulation parameters for apartment model.}
  \label{tbl:simParam2}
\end{table}

Fig.~\ref{fig:simulation} shows the
cumulative probability distribution of resulting spectral efficiencies (bits/sec/Hz)
based on 100 random drops of the network.
In this figure, IC is compared against reuse 1 where no IC is used and all links transmit
on the same bandwidth, treating interference as noise.  We see that for cell edge users
(defined as users in the lowest 10\% of rates), there is a gain of almost 4 times in the rate
by employing IC.  Hence, IC has the potential for significant value for improving rates
in such femtocellular scenarios, particularly for mobiles in strong interference.

Also shown in Fig.~\ref{fig:simulation} is a comparison of the optimal rate selection
against the selection from max-sum loopy BP run for 4 iterations.  
We see there is an exact match.  In fact,
even after only 1 iteration there is a minimal difference.  The optimal rate selection is found
as follows:  For each RX$i$, the rate pair $(R_i,R_{\sigma(i)})$ must be either in
the reuse 1 region ${\mathcal C}_{\rm reuse1}(i)$ or the joint detection region
${\mathcal C}_{\rm IC}(i)$.  Since there are two choices for each link, there are a total
of $2^n$ choices in the network.  For each of the $2^n$ choices, the total rate region is convex
and we can maximize the utility over that region.  For the optimal rate selection, we
repeat the convex optimization for every one of $2^n$ choices and select the one with the
maximum total utility.  Obviously, the max-sum loopy BP algorithm is much simpler.

\begin{figure}
\centering
\includegraphics[trim = 30mm 95mm 38mm 95mm, clip,width=0.45\textwidth]{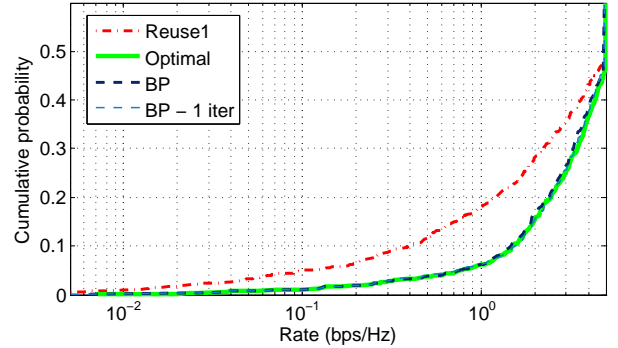}
\caption{\textbf{Simulation with 3GPP apartment model:} Plotted are the distribution of rates
(bps/Hz) over 100 drops.
We see that IC provides significant gains particularly at the cell edge.  Moreover, even one iteration
of max-sum loopy BP finds near-optimal rates for the non-convex optimization problem.}
\label{fig:simulation1}
\end{figure}

\subsection{Road Network with Mobility}

\begin{figure}
\centering
\includegraphics[width=0.35\textwidth]{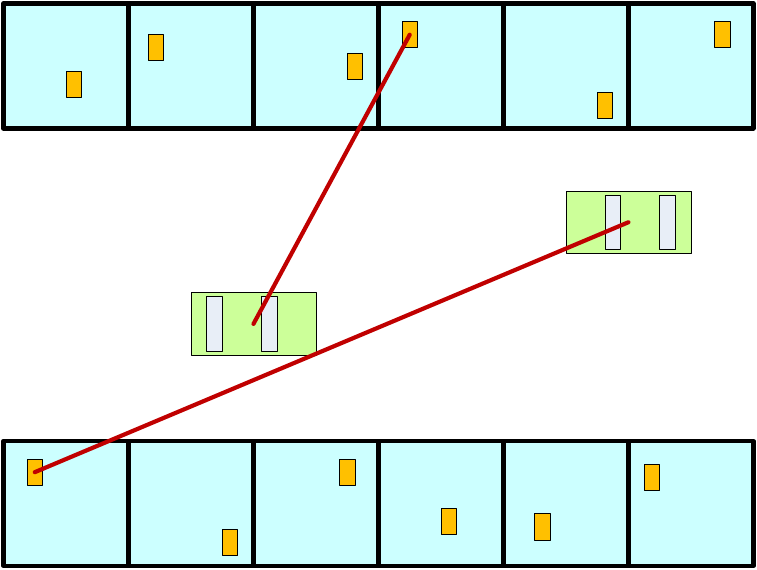}
\caption{\textbf{Road network model:}  Using open-access femtocells within apartments
may enable low-cost wide-area cellular coverage.  However, mobility within small cells
can result in significant interference as handover delays cause mobile users to drag into
neighboring cells. }
\label{fig:roadNet}
\end{figure}

A second case of strong interference in femto- or picocellular networks is handover delays.
When cells are small, handovers are frequent.  Also, since femtocells are not be directly
into the operator's core network \cite{3GPP25467},
the handovers between femtocells or from the femtocells to the macrocell
may be significantly delayed \cite{AndrewsCDRC:12}.
Due to the delays, the mobile may drag significantly into cells other than the serving cell,
exposing the mobile to strong interference.

To evaluate the ability of IC to mitigate this strong interference, we considered
a road network model shown in Fig.~\ref{fig:roadNet}.  A similar model was considered in
\cite{RanganErkip:11}.
In this model, the transmitters are femtocells placed in apartments at two sides of
the road. The receivers are mobile on the road with a random velocity between 15 and 25 m/s.
Each mobile is initially connected to the strongest serving cell.
After the connection, we let mobiles to move for a time uniformly distributed between 0 and 1 second
to model the effect of delayed handover.
We then perform a static simulation taken at this snapshot in time.
The complete list of simulation parameters are given on the Table~\ref{tbl:simParam2},
which are loosely based on the path loss models in \cite{FemtoForum:10}.

\begin{table}
  \begin{center}
    \begin{tabular}{|p{1in}|p{2in}|}
      \hline
      Parameter & Value \\ \hline
      Network topology & 10 Tx's in/on apartments by the road, 10 mobile Rx's on the road in periodic\\ \hline
      Road width & 10 m\\ \hline
      Road length & 50 m periodicity\\ \hline
      Apartment dim. & 10 m width and 20 m length\\ \hline
      Bandwidth & 5 MHz \\ \hline
      Wall loss & 0 dB \\ \hline
      Lognormal shadowing & 10 dB std.\ dev.\ \\ \hline
      Path loss & $10 + 37\log_{10}(R)$ dB, $R$ distance in meters. \\ \hline
      Rx velocity & 15 - 25 m/s \\ \hline
      Femto BS TX power & 0 dBm \\ \hline
      Femto UE noise figure & 4 dB \\ \hline
    \end{tabular}
  \end{center}
  \caption{Simulation parameters for road model.}
  \label{tbl:simParam1}
\end{table}

Fig.~\ref{fig:simulation} shows the distribution of rates based on 100 random drops.
Again, we see that IC provides significant improvement in cell edge rate:
the 10\% rate is increased by more than a factor of 3.  In fact, even the median
rate is increased almost twofold.  In addition, max-sum loopy BP finds near optimal rates.

\begin{figure}
\centering
\includegraphics[trim = 30mm 95mm 38mm 95mm, clip,width=0.45\textwidth]{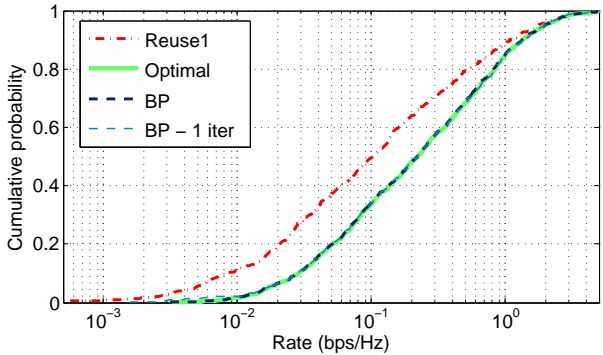}
\caption{\textbf{Simulation for the road network model:}
Plotted are the distribution of rates (bps/Hz) over 100 drops.
Again, we see that IC provides significant gains in rates throughout the rates,
but particularly at the cell edge.  Moreover,
max-sum loopy BP finds near-optimal selections after only 1 iteration.}
\label{fig:simulation}
\end{figure}

\section{Conclusions}

A general methodology based on max-sum loopy BP
is presented for wireless scheduling and resource allocation problems
in networks where each link has at most one dominant interferer.
The algorithm is extremely general and can incorporate advanced link-layer mechanisms
available in modern cellular systems.  Moreover, the method is completely distributed,
requires low communication overhead and minimal computation.
In addition, although we did not establish the convergence of the algorithm,
we have shown that if it converges, the result is guaranteed to be globally optimal.
Remarkably, this result holds for arbitrary problems, even when the problems are non-convex,
the only requirement being the dominant interferer assumption.

The method was applied to systems with IC.  We propose that when IC is available at the link-layer,
rate control as opposed to power control can provide an effective mechanism for interference
coordination.  Moreover, the optimal rate selection is performed extremely well by
max-sum loopy BP.  Simulations of the algorithm in femtocellular networks demonstrate
considerable improvement in capacity from the addition of IC for handling both restricted
association and mobility.
In addition, loopy BP appears to
provide near optimal performance with only one or two rounds of messaging, thus making the methodology
extremely attractive for practical networks.

\ifthenelse{\boolean{conference}}{
\section*{Acknowledgements}
This material is based upon work supported by the National Science
Foundation under Grant No. 1116589.
}{
}

\appendices
\section{Proof of Theorem \ref{thm:fixPoint}} \label{sec:proof}

We begin by reviewing the following well-known result from \cite{Weiss:00,WeissFree:01}:
Let $G_U = (V,E_U)$ be an undirected graph with vertices $V = \{1,\ldots,n\}$, and suppose
that a function $F(\xbf)$ admits a pairwise factorization of the form
\beq \label{eq:Fpair}
    F(\xbf) = \sum_{(i,j)\in E_U}f_{ij}(x_i,x_j),
\eeq
for some set of functions $f_{ij}(x_i,x_j)$.  Note that the graph $G_U$ is \emph{not} the factor graph.
One can then apply max-sum loopy BP to attempt to perform the optimization
\beq \label{eq:xoptPair}
    \xbfhat = \argmax_{\xbf} F(\xbf).
\eeq
It is well-known that max-sum loopy BP will provide the optimal solution when $G_U$ has no cycles.
However, our analysis requires the following.

\medskip
\begin{theorem}[\!\cite{WeissFree:01}]  \label{thm:SLT}
Consider max-sum loopy BP applied to the optimization
\eqref{eq:xoptPair} to an arbitrary
pairwise objective function of the form \eqref{eq:Fpair}.  Suppose that each
connected component of the undirected graph $G_U$ has at most one cycle.
Then, if the belief messages are fixed points of max-sum loopy BP, the corresponding
solution $\xbfhat$ is optimal in that $F(\xbfhat) \geq F(\xbf)$ for all
vectors $\xbf$.
\end{theorem}

\medskip

Now, the function \eqref{eq:Fsum} we are interested
in is a special case of the form \eqref{eq:Fpair}, so we need to show that the
corresponding undirected graph, $G_U$, has the same single cycle property.
The undirected graph for the objective function \eqref{eq:Fsum} is simple to describe:
First let $G_D=(V,E_D)$ be the \emph{directed} graph with vertices $V=\{1,\ldots,n\}$ and
directed edges
\[
    E_D = \{(\sigma(i),i), ~i=1,\ldots,n \}.
\]
Then, if $G_U = (V,E_U)$ is the corresponding undirected graph to $G_D$, then it is also
the undirected graph for the objective function \eqref{eq:Fsum}.
For illustration, Fig.~\ref{fig:netExample1} shows the directed graph $G_D$ for the
example network mentioned earlier in Fig.~\ref{fig:netExample}.
Note that there is a directed edge from $j$ to $i$ if and only if TX$j$ is the dominant
interferer of RX$i$.  Since each receiver has at most one dominant interferer,
each node in $G_D$ has at most one incoming edge.  That is, the in-degree is at most one for all
nodes.  This property will be key.

\begin{figure}
\centering
\includegraphics[width=0.2\textwidth]{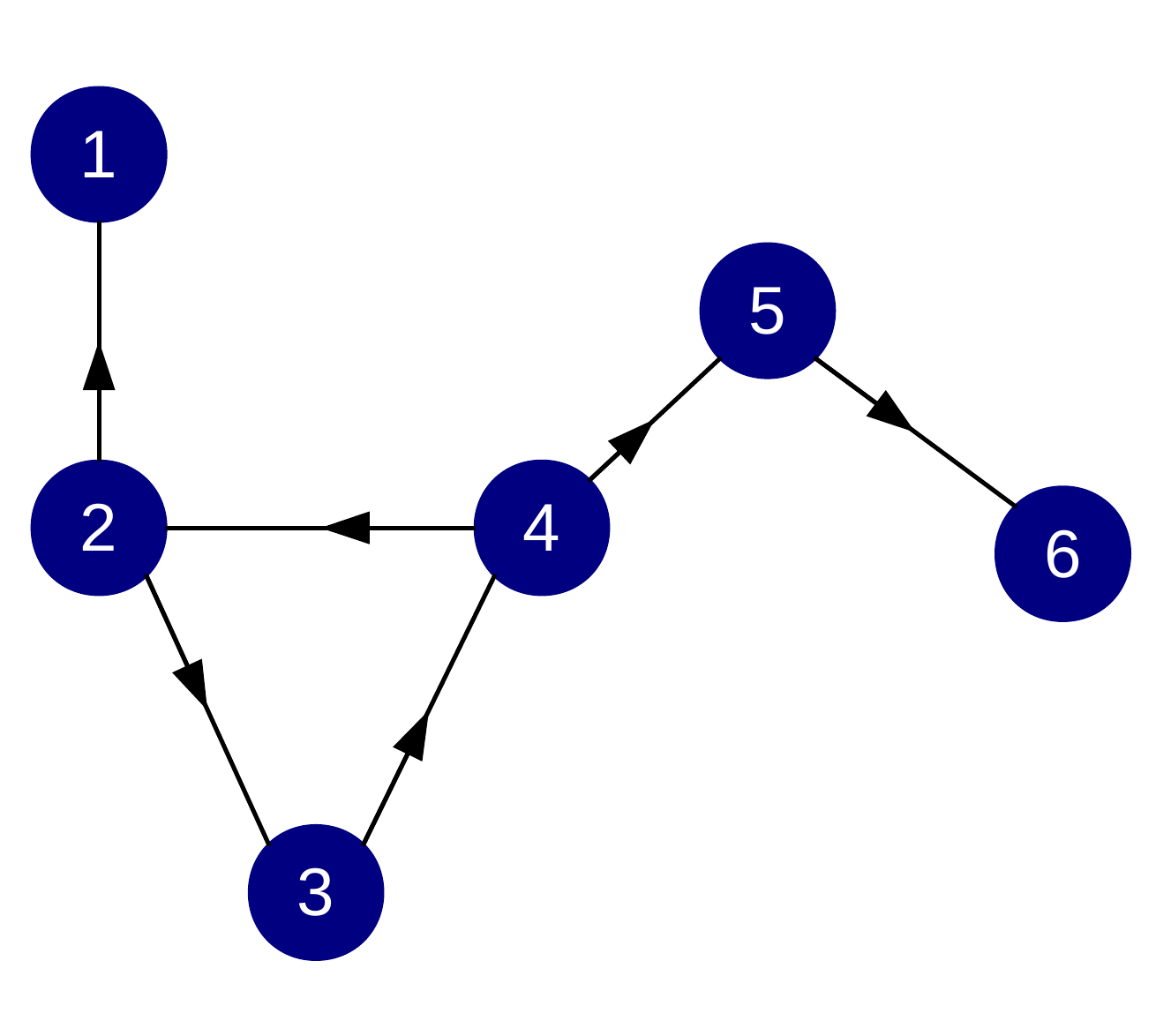}
\caption{\textbf{Directed graph} for the example network in Fig.~\ref{fig:netExample}.
There is a directed edge from $j$ to $i$ if and only if $j=\sigma(i)$, namely
TX$j$ is the dominant interferer of RX$i$. }
\label{fig:netExample1}
\end{figure}

To apply Theorem \ref{thm:SLT}, we need to show that each connected
component of $G_U$ has at most one cycle.
We begin with the following lemma.

\begin{lemma} \label{lem:path} Consider any path $P_M = (i_1,i_2,...,i_M)$,
in the undirected graph $G_U$.
Then either $(i_2,i_1) \in E_D$ or
$(i_{M-1},i_M) \in E_D$.  That is, either the first or last segments of the path are pointing
``outwards".
\end{lemma}
\begin{proof}  The proof follows by induction. For $M=2$, either
$(i_2,i_1)$ or $(i_1,i_2)$ must be a directed edge.  Now suppose the the lemma is true
for all paths of length $M-1$, and consider a path of length $M$.
If $(i_2,i_1) \in E_D$ we are done, so suppose $(i_2,i_1) \not \in E_D$.
Then, applying the induction hypothesis to $P_{M-1}$ shows that $(i_{M-2},i_{M-1})$ must be a
directed edge.  But, since the node $i_{M-1}$ can have at most one incoming edge,
$(i_{M-1},i_{M})$ must be outgoing and hence in $E_D$.  So, we have shown that either
$(i_2,i_1) \in E_D$ or $(i_{M-1},i_M) \in E_D$.
\end{proof}

\medskip

We wow return to the graph $G_U$ and show that each connected component can have at most
one cycle.  We prove the property by contradiction.
Assume there are two loops in the same connected component of $G_U$:
\beqan
    L &=& (i_1,i_2,...,i_K,i_1)\\
    J &=& (j_1,j_2,...,j_M,j_1),
\eeqan
with no repetitions of points within each loop, aside from the first and last points.
Now there are three possibilities as shown in Fig.~\ref{fig:graph}:  (a) The two loops
have no common point, (b) one common point and (c) two or more common points.  We show each of these
possibilities leads to a contradiction.
\begin{figure}
\centering
\includegraphics[width=0.2\textwidth]{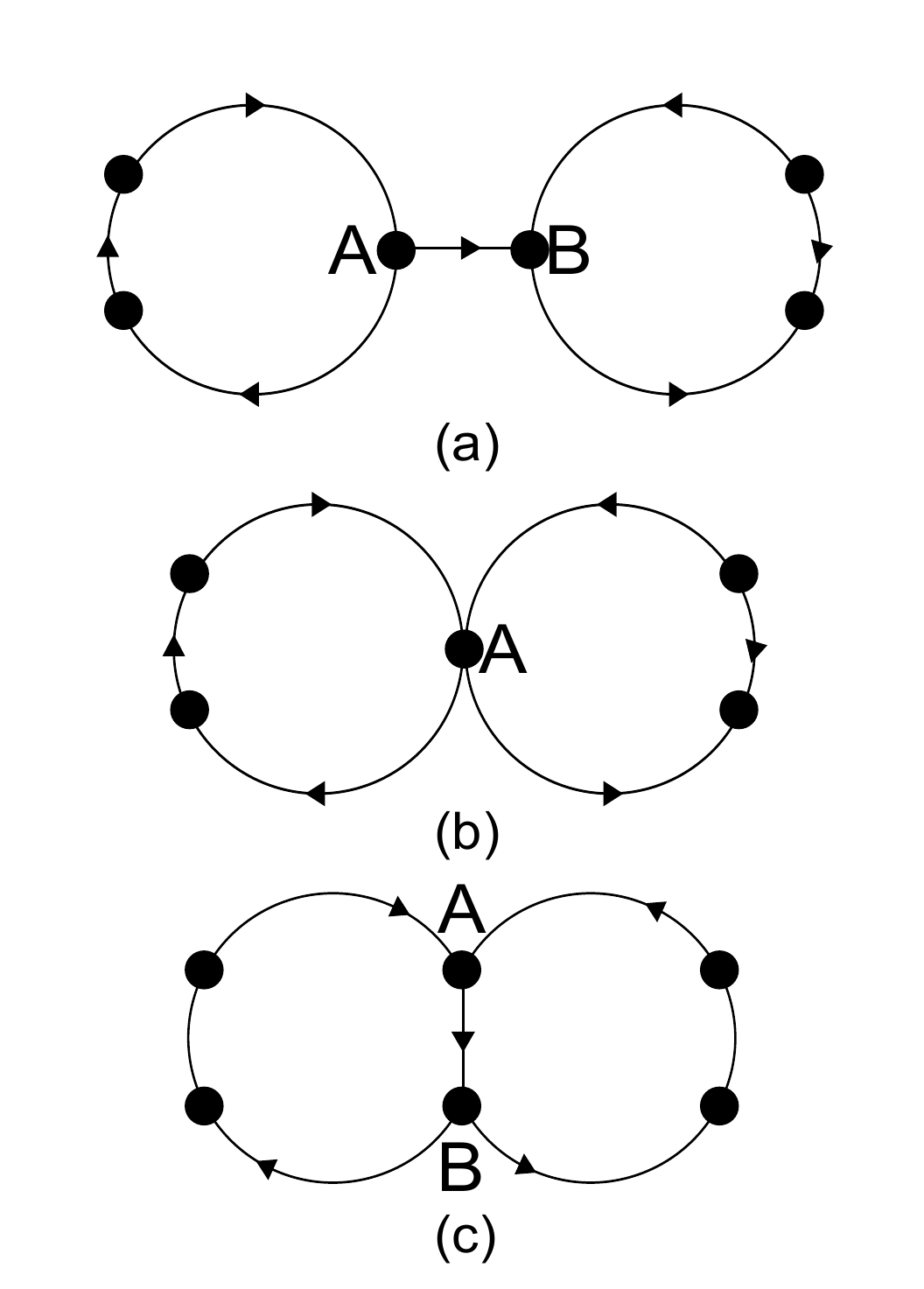}
\caption{Possible positions of two loops in the same connected component.}
\label{fig:graph}
\end{figure}

\paragraph*{Case (a) No common points}  By assumption, $L$ and $J$ are in the same
connected component, so there must be a path in the undirected graph
between between one point in $L$ and a second point
in $J$.  Without loss of generality (WLOG),
we can assume the path goes from $i_1$ to $j_1$ and we will write the
path as
\[
    P = (i_1=k_1, k_2, \ldots, k_N = j_1).
\]
We can construct this path such that, aside from $i_1$ and $j_1$, this path has no
common points with the loops $L$  or $J$.
Now, by Lemma \ref{lem:path}, either $(k_2,k_1)$ or $(k_{N-1},k_N)$ is a directed edge (i.e.\ in $E_D$).
WLOG suppose $(k_2,k_1) \in E_D$; the other case is similar.
The point $k_1$ would correspond to point B in Fig.~\ref{fig:graph} with the loop on
the right being loop $L$.
Thus $k_1=i_1$ has an incoming edge
from $k_2$.  But Lemma \ref{lem:path} applied to the path $L$ implies that there must
be an incoming edge from either $i_2$ or $i_{K}$.  When combined with the incoming edge from
$k_2$, $i_1$ has more than one incoming edge, which is impossible.

\paragraph*{Case (b) One common point} WLOG we can assume that the point $i_1=j_1$ is the
common point in loops $L$ and $J$.  The point is shown as A in Fig.~\ref{fig:graph}.
Applying Lemma \ref{lem:path} to both loops $L$ and $J$ show that there must be at least
one incoming edge into point A from each loop.  Thus, A has more than one incoming edge,
which again is a contradiction.

\paragraph*{Case (c) Two or more common points}
Taking a minimal common path, and possibly rearranging
the indexing, we can assume that the first $N$ points of the loops $L$ and $J$ are common so that
\[
    i_1=j_1, \ldots, i_N=j_N,
\]
but $i_{N+1} \neq j_{N+1}$, and $i_{K} \neq j_{M}$.
The situation is depicted in Fig.~\ref{fig:graph}, where point A corresponds to $i_1$ ($=j_1$)
and B corresponds to $i_N$ ($=j_N$).
Applying Lemma \ref{lem:path} to the path along the common segment from A to B,
\[
    P = (i_1, \ldots, i_N)=(j_1, \ldots, j_N),
\]
implies that either $(i_2,i_1)$ or $(i_{N-1},i_N)$ is a directed edge.  WLOG assume
$(i_{N-1},i_N)$ is a directed edge, creating an incoming edge into point B
(the point with index $i_N=j_N$).  Since any node can have at most one incoming edge, both the edges
$(i_N,i_{N+1})$ in loop $L$ and $(j_N,j_{N+1})$ in loop $J$ must be directed edges
outgoing from $i_N=j_N$.
Now consider the two paths from $i_N=j_N$ back to $i_1=j_1$ along loops $L$
and $J$.  From Lemma \ref{lem:path},
since the initial segments, $(i_N,i_{N+1})$ in loop $L$ and $(j_N,j_{N+1})$ in loop $J$,
are directed edges, the final segments $(i_K,i_1)$  and $(j_M,j_1)$ must also be
directed edges.  Thus, there are two incoming edges into point A
(the point with index $i_1=j_1$).
This is a contradiction.

\medskip
Thus, all three cases are impossible, and we have proven by contradiction that
any connected component of $G_U$ can have at most one cycle.  We can therefore
apply Theorem \ref{thm:SLT} and the proof of Theorem \ref{thm:fixPoint} is complete.

%
%




%

\bibliographystyle{IEEEtran}
\bibliography{bibl}

\newcommand{\SortNoop}[1]{}
\begin{thebibliography}{10}
\providecommand{\url}[1]{#1}
\csname url@samestyle\endcsname
\providecommand{\newblock}{\relax}
\providecommand{\bibinfo}[2]{#2}
\providecommand{\BIBentrySTDinterwordspacing}{\spaceskip=0pt\relax}
\providecommand{\BIBentryALTinterwordstretchfactor}{4}
\providecommand{\BIBentryALTinterwordspacing}{\spaceskip=\fontdimen2\font plus
\BIBentryALTinterwordstretchfactor\fontdimen3\font minus
  \fontdimen4\font\relax}
\providecommand{\BIBforeignlanguage}[2]{{%
\expandafter\ifx\csname l@#1\endcsname\relax
\typeout{** WARNING: IEEEtran.bst: No hyphenation pattern has been}%
\typeout{** loaded for the language `#1'. Using the pattern for}%
\typeout{** the default language instead.}%
\else
\language=\csname l@#1\endcsname
\fi
#2}}
\providecommand{\BIBdecl}{\relax}
\BIBdecl

\bibitem{ChaAndG:08}
V.~Chandrasekhar, J.~G. Andrews, and A.~Gatherer, ``Femtocell networks: A
  survey,'' \emph{IEEE Comm. Mag.}, vol.~46, no.~9, pp. 59--67, Sep. 2009.

\bibitem{LopezVRZ:09}
D.~L{\'o}pez-P{\'e}rez, A.~Valcarce, G.~de~la Roche, and J.~Zhang, ``{OFDMA}
  femtocells: A roadmap on interference avoidance,'' \emph{IEEE Comm. Mag.},
  vol.~47, no.~9, pp. 41--48, Sep. 2009.

\bibitem{AndrewsCDRC:12}
J.~G. Andrews, H.~Claussen, M.~Dohler, S.~Rangan, and M.~C. Reed, ``Femtocells:
  Past, present, and future,'' \emph{IEEE J. Sel. Areas Comm.}, vol.~30, no.~3,
  Apr. 2012.

\bibitem{3GPPICIC}
3GPP, ``{New Work Item Proposal: Enhanced {ICIC} for non-{CA} based deployments
  of heterogeneous networks for LTE},'' RP-100372, 2010.

\bibitem{FemtoForum:10}
{Femto Forum}, ``Interference management in {OFDMA} femtocells,'' Whitepaper
  available at www.femtoforum.org, Mar. 2010.

\bibitem{WainwrightJ:08}
M.~J. Wainwright and M.~I. Jordan, \emph{Graphical Models, Exponential
  Families, and Variational Inference}, ser. Foundations and Trends in Machine
  Learning.\hskip 1em plus 0.5em minus 0.4em\relax Hanover, MA: NOW Publishers,
  2008, vol.~1.

\bibitem{Frey:98}
B.~J. Frey, \emph{Graphical Models for Machine Learning and Digital
  Communication}.\hskip 1em plus 0.5em minus 0.4em\relax MIT Press, 1998.

\bibitem{Chiang:02}
M.~Chiang, ``Distributed network control through sum product algorithm on
  graphs,'' in \emph{Proc.\ IEEE Globecom}, vol.~3, Nov. 2002, pp. 2395 --
  2399.

\bibitem{SanghaviMW:07}
S.~Sanghavi, D.~Malioutov, and A.~Willsky, ``Belief propagation and {LP}
  relaxation for weighted matching in general graphs,'' in \emph{Proc.\ NIPS},
  December 2007.

\bibitem{BayatiSS:08}
M.~Bayati, D.~Shah, and M.~Sharma, ``Max-product for maximum weight matching:
  convergence, correctness and {LP} duality,'' \emph{IEEE Trans. Inform.
  Theory}, vol.~54, no.~3, pp. 1241--1251, March 2008.

\bibitem{sohn_2010}
I.~Sohn, S.~H. Lee, and J.~G. Andrews, ``A graphical model approach to downlink
  cooperative {MIMO} systems,'' \emph{IEEE Globecom}, 2010.

\bibitem{RanganM:11}
S.~Rangan and R.~Madan, ``Belief propagation methods for intercell interference
  coordination,'' in \emph{Proc. IEEE Infocom}, Shanghai, China, Apr. 2011.

\bibitem{RanganM:12}
------, ``Belief propagation methods for intercell interference coordination in
  femtocell networks,'' \emph{IEEE J. Sel. Areas Comm.}, vol.~30, no.~3, pp.
  631--640, Apr. 2012.

\bibitem{RanganFGS:12-ISIT}
S.~Rangan, A.~K. Fletcher, V.~K. Goyal, and P.~Schniter, ``Hybrid generalized
  approximation message passing with applications to structured sparsity,'' in
  \emph{Proc. IEEE Int. Symp. Inform. Theory}, Cambridge, MA, Jul. 2012, pp.
  1241--1245.

\bibitem{Dahlman:07}
E.~Dahlman, S.~Parkvall, J.~Sk{\"o}ld, and P.~Beming, \emph{{3G} Evolution:
  {HSPA} and {LTE} for Mobile Broadband}.\hskip 1em plus 0.5em minus
  0.4em\relax Oxford, UK: Academic Press, 2007.

\bibitem{TassiulasE:92}
L.~Tassiulas and A.~Ephremides, ``Stability properties of constrained queueing
  systems and scheduling policies for maximum throughput in multihop radio
  networks,'' \emph{IEEE Trans. Automat. Control}, vol.~37, pp. 1936--1948,
  1992.

\bibitem{stolyar_2005}
A.~Stolyar, ``On the asymptotic optimality of the gradient scheduling algorithm
  for multi-user throughput allocation,'' \emph{Oper. Res.}, 2005.

\bibitem{WeissFree:01}
Y.~Weiss and W.T.Freeman, ``{On the optimality of solutions of the max-product
  belief-propagation algorithm in arbitrary graphs},'' \emph{IEEE Trans.
  Inform. Theory}, vol.~47, no.~2, pp. 736--744, Feb. 2001.

\bibitem{Ahlswede:71}
R.~Ahlswede, ``Multi-way communication channels,'' in \emph{Proc. IEEE Int.
  Symp. Inform. Theory}, Armenian S.S.R., Sep. 1971, pp. 23--52.

\bibitem{CoverT:91}
T.~M. Cover and J.~A. Thomas, \emph{Elements of Information Theory}.\hskip 1em
  plus 0.5em minus 0.4em\relax New York: John Wiley \& Sons, 1991.

\bibitem{HanKob:81}
T.~Han and K.~Kobayashi, ``{A new achievable rate region for the interference
  channel},'' \emph{IEEE Trans. Inform. Theory}, vol.~27, no.~1, pp. 49--60,
  Jan. 1981.

\bibitem{EtkinTseWa:08}
R.~Etkin, D.~Tse, and H.~Wang, ``Gaussian interference channel capacity to
  within one bit,'' \emph{IEEE Trans. Inform. Theory}, vol.~54, no.~12, pp.
  5534 --5562, Dec. 2008.

\bibitem{Rangan:10-GC}
S.~Rangan, ``Femto-macro cellular interference control with subband scheduling
  and interference cancelation,'' in \emph{Proc.\ IEEE Globecom}, Miami, FL,
  Dec. 2010, pp. 695--700.

\bibitem{Andrews:05}
J.~G. Andrews, ``Interference cancellation for cellular systems: A contemporary
  overview,'' \emph{IEEE Wireless Comm.}, vol.~12, no.~2, pp. 19--29, Apr.
  2005.

\bibitem{BourdreauPGCWV:09}
G.~Boudreau, J.~Panicker, N.~Guo, R.~Chang, N.~Wang, and S.~Vrzic,
  ``Interference coordination and cancellation for {4G} networks,'' \emph{IEEE
  Comm. Mag.}, vol.~47, no.~4, pp. 74--81, Apr. 2009.

\bibitem{ShiReed:07}
Z.~Shi and M.~C. Reed, ``{Iterative maximal ratio combining channel estimation
  for multiuser detection on a time frequency selective wireless CDMA
  channel},'' in \emph{Proc.\ IEEE Wireless Communications and Networking
  Conference}, Hong Kong, Mar. 2007.

\bibitem{ChiangHLT:08}
M.~Chiang, P.~Hande, T.~Lan, and C.~W. Tan, ``Power control in wireless
  cellular networks,'' \emph{Foundation and Trends in Networking}, vol.~2,
  no.~4, pp. 381--533, Jul. 2008.

\bibitem{KelleyMT:98}
F.~P. Kelly, A.~K. Maulloo, and D.~K.~H. Tan, ``Rate control for communication
  networks: shadow prices, proportional fairness and stability,'' \emph{Journal
  of the Operational Research Society}, vol.~49, no.~3, pp. 237--252, Mar.
  1998.

\bibitem{ShakkottaiS:07}
S.~Shakkottai and R.~Srikant, \emph{Network Optimization and Control}, ser.
  Foundations and Trends in Networking.\hskip 1em plus 0.5em minus 0.4em\relax
  NOW Publishers, 2007.

\bibitem{FlashLinq:10-allerton}
S.~Tavildar, S.~Shakkottai, T.~Richardson, J.~Li, R.~Laroia, and A.~Jovicic,
  ``Flash{L}in{Q}: A synchronous distributed scheduler for peer-to-peer ad hoc
  networks,'' in \emph{Proc.\ Allerton Conf.\ Comm.\ Control \& Comp.},
  Allerton, IL, Oct. 2010.

\bibitem{NiSrikant:09}
J.~Ni and R.~Srikant, ``Distributed {CSMA/CA} algorithms for achieving maximum
  throughput in wireless networks,'' in \emph{Proc.\ Information Theory and
  Applications Workshop}, San Diego, CA, Jan. 2009.

\bibitem{3GPP36.211}
3GPP, ``{E}volved {U}niversal {T}errestrial {R}adio {A}ccess (e-utra);
  {P}hysical {C}hannels and {M}odulation,'' TS 36.211 (release 10), 2012.

\bibitem{3GPP25467}
------, ``{UTRAN} architecture for {3G Home Node B} {(HNB)}; stage 2,'' TS
  25.467 (release 9), 2010.

\bibitem{RanganErkip:11}
S.~Rangan and E.~Erkip, ``{Hierarchical Mobility in Dense Cellular Networks via
  Relaying},'' in \emph{Proc.\ Globecomm}, Houston, TX, Dec. 2011.

\bibitem{Weiss:00}
Y.~Weiss, ``{Correctness of Local Probability Propagation in Graphical Models
  with Loops},'' \emph{Neural Comp.}, vol.~12, no.~1, pp. 1--41, Jan. 2000.

\end{thebibliography}

%
%








\end{document}